\documentclass[a4,11pt]{article}
\usepackage[margin=1in]{geometry}

\newtheorem{lemma}{Lemma}
\newtheorem{theorem}{Theorem}%[section]
\newenvironment{proof}%
{\begin{trivlist}\item[\hspace*{\labelsep}{\it Proof.\/}]}%
{\hfill$\Box$\end{trivlist}}

\newcommand{\head}[1]
 {\markright{\hbox to 0pt{\vtop to 0pt{\hbox{}\vskip 3mm \hrule
 width  \textwidth \vss} \hss}{\sc #1}}}
\usepackage{latexsym,graphicx,amsfonts,amsmath}
\begin{document}

%%\begin{titlepage}
\title{\bf Approximating the optimal competitive ratio for an ancient online scheduling problem}

\author{Lin Chen$^1$\ \  Deshi Ye$^1$\ \ Guochuan
Zhang$^1$
\\{\small $^1$College of Computer Science, Zhejiang University, Hangzhou, 310027, China
}
\\{\small chenlin198662@zju.edu.cn, zgc@zju.edu.cn}
 }
\date{}
\maketitle

\begin{abstract}
We consider the classical online scheduling problem $P||C_{max}$ in
which jobs are released over list and provide a nearly optimal
online algorithm. More precisely, an online algorithm whose
competitive ratio is at most $(1+\epsilon)$ times that of an optimal
online algorithm could be achieved in polynomial time, where $m$,
the number of machines, is a part of the input. It substantially
improves upon the previous results by almost closing the gap between
the currently best known lower bound of 1.88~\cite{Rudin01} and the
best known upper bound of 1.92~\cite{FleWah00}. It has been known by
folklore that an online problem could be viewed as a game between an
adversary and the online player. Our approach extensively explores
such a structure and builds up a completely new framework to show
that, for the online over list scheduling problem, given any
$\epsilon>0$, there exists a uniform threshold $K$ which is
polynomial in $m$ such that if the competitive ratio of an online
algorithm is $\rho\le 2$, then there exists a list of at most $K$
jobs to enforce the online algorithm to achieve a competitive ratio
of at least $\rho-O(\epsilon)$. Our approach is substantially
different from that of~\cite{Gunther2013soda}, in which an
approximation scheme for online over time scheduling problems is
given, where the number of machines is fixed. Our method could also
be extended to several related online over list scheduling models.

\bigskip
\smallskip\noindent{\bf Keywords:} {Competitive analysis; Online
scheduling; Dynamic programming.}
\end{abstract}

\baselineskip 15pt
\section{Introduction}
Very recently G\"{u}nther et al.~\cite{Gunther2013soda} come up with
a nice notion called {\em Competitive ratio approximation scheme}
for online problems. Formally speaking, it is a series of online
algorithms $\{A_{\epsilon}:\epsilon>0\}$, where $A_{\epsilon}$ has a
competitive ratio at most $(1+\epsilon)$ times the optimal
competitive ratio. Naturally, a competitive ratio approximation
scheme could be seen as an online version of the PTAS (polynomial
time approximation scheme) for the offline problems. Using such a
notion, they provide nearly optimal online algorithms for several
online scheduling problems where jobs arrive over time, including
$Qm|r_j,(pmtn)|\sum w_jc_j$ as well as $Pm|r_j|C_{max}$, where $m$
is the number of machines. The algorithm runs in polynomial time
when $m$ is fixed.

That is a great idea for designing nearly optimal online algorithms,
that motivates us to revisit the classical online problems which
still have a gap between upper and lower bounds. However, the
technique of G\"{u}nther et al.~\cite{Gunther2013soda} heavily
relies on the structure of the optimal solution for the over time
scheduling problem, through which they can focus on jobs released
during a time window of a constant length. It thus seems hard to
generalize to other online models.

Clearly, the first online scheduling problem which should be
revisited is $P||C_{max}$, a fundamental problem in which jobs are
released over list. This ancient scheduling model admits a simple
algorithm called $LS$ (list scheduling)~\cite{Graham66}. Its
competitive ratio is $2-1/m$ that achieves the best possible for
$m=2,3$~\cite{FaKeTu89}. Nevertheless, better algorithms exist for
$m=4,5,6,7$, see~\cite{ChVlWo94A}~\cite{GalWoe93B}~\cite{RudCha03}
for upper and lower bounds for online scheduling problems where $m$
taking these specified values. Many more attentions are paid to the
general case where $m$ is arbitrary. There is a long list of
improvements on upper and lower bounds,
see~\cite{Albers99}~\cite{BaFiKV95}~\cite{KaPhTo94} for improvements
on competitive algorithms,
and~\cite{Albers99}~\cite{BaKaRa94}~\cite{gormley2000lb} for
improvements on lower bounds. Among them the currently best known
upper bound is $1+\sqrt{\frac{1+\ln2}{2}} \approx
1.9201$~\cite{FleWah00}, while the best known lower bound is
1.88~\cite{Rudin01}. We refer the readers to~\cite{Sgall98} for a
nice survey on this topic.

Although the gap between the upper and lower bounds are relatively
small, it leaves a great challenge to close it. In this paper we
tackle this classical problem by providing a competitive ratio
approximation scheme. The running time is polynomial in the input
size. More precisely, the time complexity related to $m$ is
$O(m^{\Lambda})$ where $\Lambda=2^{O(1/\epsilon^2\log^2
(1/\epsilon))}$. It is thus polynomial even when the number of
machines is a part of the input.

To simplify the notion, throughout this paper we use {\em
competitive scheme} instead of competitive ratio approximation
scheme.

\paragraph{\bf General Ideas} We try to give a full picture of
our techniques. Given any $\epsilon>0$, at any time it is possible
to choose a proper value (called a scaling factor) and scale all the
jobs released so far such that there are only a constant number of
different kinds of jobs. We then represent the jobs scheduled on
each machine by a tuple (called a trimmed-state) in which the number
of each kind of jobs remains unchanged. Composing the trimmed-states
of all machines forms a trimmed-scenario and the number of different
trimmed-scenarios we need to consider is a polynomial in $m$,
subject to the scaling factors.

Given a trimmed-scenario, we can compute the corresponding
approximation ratio (comparing with the optimal schedule), which is
called an {\em instant approximation ratio}. Specifically, if the
schedule arrives at a trimmed-scenario when the adversary stops,
then the competitive ratio equals to the instant approximation ratio
of this trimmed-scenario. Formal definitions will be given in the
next section. Note that the instant approximation ratio of every
trimmed-scenario could be determined (up to an error of
$O(\epsilon)$) regardless of the scaling factor.

To understand our approach easily we consider the online scheduling
problem as a game. Each time the adversary and the scheduler take a
move, alternatively, i.e., the adversary releases a job and the
online scheduler then assigns the job to a machine. It transfers the
current trimmed-scenario into a new one. Suppose the adversary wins
the game by leading it into a certain trimmed-scenario with an
instant approximation ratio $\rho$, forcing the competitive ratio to
be at least $\rho$. The key observation is that if he has a winning
strategy, he would have a winning strategy of taking only a
polynomial number (in $m$) of moves since the game itself consists
of only a polynomial number of distinct trimmed-scenarios. A
rigorous proof for such an observation relies on formulating the
game into a layered graph and associating the scheduling of any
online algorithm with a path in it. Given the observation, the
online problem asks if the adversary has a winning strategy of
$C=poly(m)$ moves, starting from a trimmed-scenario where there is
no job. Such a problem could be solved via dynamic programming,
which decomposes it into a series of subproblems that ask whether
the adversary has a winning strategy of $C'<C$ moves, starting from
an arbitrary trimmed-scenario.

Various extensions could be built upon this framework. Indeed,
competitive schemes could be achieved for $Rm||C_{max}$ and
$Rm||\sum_i C_i^p$ where $p\ge 1$ is some constant and $C_i$ is the
completion time of machine $i$. The running times of these schemes
are polynomial when $m$ is a constant.

In addition to competitive schemes, it is interesting to ask if we
can achieve an optimal online algorithm. We consider the semi-online
model $P|p_j\le q|C_{max}$, where all job processing times are
bounded. We are able to design an optimal online algorithm running
in $(mq)^{O(mq)}$ time. It is exponential in both $m$ and $q$.

Recall that the competitive ratio of list scheduling for
$P||C_{max}$ is $2-1/m$. Throughout the paper we focus on online
algorithms whose competitive ratio is no more than $2$. We assume
that $m\ge 2$.

\section{Structuring Instances}
To tackle the online scheduling problem, similarly as the offline
case we want to well structure the input instance subject to an
arbitrarily small loss. However, in the online setting we are not
aware of the whole input. The instance needs scaling in a dynamic
way.

Given any $0<\epsilon\le 1/4$, we may assume that all the jobs
released have a processing time of $(1+\epsilon)^j$ for some integer
$j\ge 0$. Let $c_0$ be the smallest integer such that
$(1+\epsilon)^{c_0}\ge 1/\epsilon$. Let $\omega$ be the smallest
integer such that $(1+\epsilon)^{\omega}\ge 3$. Let
$SC=\{(1+\epsilon)^{j\omega}|j\ge 0, j\in \mathbb{N}\}$.

Consider the schedule of $n$ ($n\ge 1$) jobs by any online
algorithm. Let $p_{max}=\max_j\{p_j\}$. Then
$LB=\max\{\sum_{j=1}^np_j/m, p_{max}\}$ is a trivial lower bound on
the makespan. We choose $T_{LB}\in SC$ such that $T_{LB}\le
LB<T_{LB}(1+\epsilon)^{\omega}$, and define job $j$ as a small job
if $p_j\le T_{LB}(1+\epsilon)^{-c_0}$, and a big job otherwise.
$T_{LB}$ is called the {\em scaling factor} of this schedule.

Let $L_h^s$ be the load (total processing time) of small jobs on
machine $h$. An $(\omega+c_0+1)$-tuple
$st_h=(\eta_{-c_0}^h,\eta_{-c_0+1}^h,\cdots,\eta_{\omega}^h)$ is
used to represent the jobs scheduled on machine $h$, where
$\eta_i^h$ ($-c_0+1\le i\le \omega$) is the number of big jobs with
processing time $T_{LB}(1+\epsilon)^{i}$ on machine $h$, and
$\eta_{-c_0}^h=L^h_s/(T_{LB}(1+\epsilon)^{-c_0})$. We call such a
tuple as a state (of machine $h$). The first coordinate of a state
might be fractional, while the other coordinates are integers. The
load of a state is defined as
$LD(st_h)=\sum_{i=-c_0}^{\omega}(1+\epsilon)^i\eta_i\le 4LB$.

Composing the states of all machines forms a scenario
$\psi=(st_1,st_2,\cdots,st_m)$. Thus, any schedule could be
represented by $(T_{LB},\psi)$ where $T_{LB}\in SC$ is the scaling
factor of the schedule. Specifically, if the adversary stops now,
then the competitive ratio of such a schedule is approximately (up
to an error of $O(\epsilon)$):
$$\rho(\psi)={C_{max}(\psi)}/{OPT(\psi)}$$
where $C_{max}(\psi)=\max_{j}LD(st_j)$, and $OPT(\psi)$ is the
makespan of an optimal solution for the offline scheduling problem
in which jobs of $\psi$ are taken as an input (here small jobs are
allowed to split). We define $LD(\psi)=\sum_h LD(st_h)$ and
$P_{max}(\psi)$ the largest processing time (divided by $T_{LB}$) of
jobs in $\psi$ ($P_{max}(\psi)=(1+\epsilon)^{-c_0}$ if there is no
big job in $\psi$). Obviously, $$OPT(\psi)\ge
LB=\max\{LD(\psi)/m,P_{max}(\psi)\}\ge 1. $$ The above ratio is
regardless of the scaling factor and is called an {\em instant
approximation ratio}.

We can use a slightly different $(\omega+c_0+1)$-tuple
$\tau=(\nu_{-c_0},\nu_{-c_0+1},\cdots,\nu_{\omega})$ to approximate
a state, where each coordinate is an integer. It is called a {\em
trimmed-state}. Specifically, $\tau$ is called a simulating-state of
$st_h$ if $\nu_i=\eta_i^h$ for $-c_0<i\le \omega$ and
$\eta_{-c_0}^h\le \nu_{-c_0}\le \eta_{-c_0}^h+2$.

We define $LD(\tau)=\sum_{i=-c_0}^{\omega}\nu_i(1+\epsilon)^{i}$ and
restrict our attention on trimmed-states whose load is no more than
$4LB+2(1+\epsilon)^{-c_0}$. There are at most $\Lambda\le
2^{O(1/\epsilon^2\log^2 (1/\epsilon))}$ such kinds of trimmed-states
(called feasible trimmed-states). We sort these trimmed-states
arbitrarily as $\tau_1,\cdots,\tau_{\Lambda}$, and define a
$\Lambda$-tuple $\phi=(\xi_1,\xi_2,\cdots,\xi_{\Lambda})$ to
approximate scenarios, where $\sum_i \xi_i=m$ and $0\le \xi_i\le m$
is the number of machines whose corresponding trimmed-state is
$\tau_i$. Indeed, $\phi$ is called a trimmed-scenario and
specifically, it is called a simulating-scenario of
$\psi=(st_1,st_2,\cdots,st_m)$ if there is a one to one
correspondence between the $m$ states (i.e., $st_1$ to $st_m$) and
the $m$ trimmed-states of $\phi$ such that each trimmed-state is the
simulating-state of its corresponding state.

Recall that in $\psi$, jobs are scaled with $T_{LB}$, thus $1\le
\max\{1/mLD(\psi),P_{max}(\psi)\}<(1+\epsilon)^{\omega}$. We may
restrict our attentions to trimmed-scenarios satisfying $1\le
\max\{1/mLD(\phi),P_{max}(\phi)\}<
(1+\epsilon)^{\omega}+2(1+\epsilon)^{-c_0}$, where similarly we
define $LD(\phi)=\sum_j\xi_jLD(\tau_j)$, and $P_{max}(\phi)$ the
largest processing time of jobs in $\phi$. Trimmed-scenarios
satisfying the previous inequality are called feasible
trimmed-scenarios.

Notice that there are $\Gamma\le (m+1)^{\Lambda}$ different kinds of
feasible trimmed-scenarios. we sort them as
$\phi_1,\cdots,\phi_{\Gamma}$. As an exception, we plug in two
additional trimmed-scenarios $\phi_0$ and $\phi_{\Gamma+1}$, where
$\phi_0$ represents the initial trimmed-scenario in which there are
no jobs, and $\phi_{\Gamma+1}$ represents any infeasible
trimmed-scenario. Let $\Phi$ be the set of these trimmed-scenarios.
We define $$\rho(\phi)={C_{max}(\phi)}/{OPT(\phi)}$$ as the {\em
instant approximation ratio} of a feasible trimmed-scenario $\phi$,
in which $C_{max}(\phi)=\max_{j}\{LD(\tau_j):\xi_j>0\}$, and
$OPT(\phi)$ is the makespan of the optimum solution for the offline
scheduling problem in which jobs of $\phi$ are taken as an input and
every job (including small jobs) should be scheduled integrally. As
an exception, we define $\rho(\phi_{0})=1$ and
$\rho(\phi_{\Gamma+1})=\infty$.

Furthermore, notice that except for $\phi_{\Gamma+1}$,
$C_{max}(\phi)\le 4(1+\epsilon)^{\omega}+2(1+\epsilon)^{-c_0}\le
20$, which is a constant. Thus we can divide the interval $[1,20]$
equally into $19/\epsilon$ subintervals and let
$\Delta=\{1,1+\epsilon,\cdots,1+\epsilon\cdot 19/\epsilon\}$. We
round up the instant approximation ratio of each $\phi$ to its
nearest value in $\Delta$. For simplicity, we still denote the
rounded value as $\rho(\phi)$.

\begin{lemma}
\label{le:approx-simu-scenario} If $\phi$ is a simulating-scenario
of $\psi$, then $\rho(\psi)-O(\epsilon)\le \rho(\phi)\le
\rho(\psi)+O(\epsilon)$.
\end{lemma}
\begin{proof}
It can be easily seen that $OPT(\psi)\le OPT(\phi)\le
OPT(\psi)+3(1+\epsilon)^{-c_0}$. Meanwhile $C_{max}(\psi)\le
C_{max}(\phi)\le C_{max}(\psi)+2(1+\epsilon)^{-c_0}$. Note that
$OPT(\psi)\ge 1$ and the lemma follows directly.
\end{proof}

Consider the scheduling of $n$ jobs by any online algorithm. The
whole procedure could be represented by a list as
$$(T_{LB}(1),\psi(1))\rightarrow (T_{LB}(2),\psi(2))\rightarrow\cdots \rightarrow(T_{LB}(n),\psi(n)),$$
where $\psi(k)$ is the scenario when there are $k$ jobs, and
$T_{LB}(k)$ is the corresponding scaling factor. Here $\psi(k)$
changes to $\psi(k+1)$ by adding a new job $p_{k+1}$, and the reader
may refer to Appendix~\ref{ap:add job scenario} to see how the
coordinates of a scenario change when a new job is added.

Let $\mu_0$ be the smallest integer such that
$(1+\epsilon)^{\mu_0}\ge 4(1+\epsilon)^{\omega+c_0+1}$ and $R
=\{0,(1+\epsilon)^{-c_0},\cdots,(1+\epsilon)^{\lceil
\mu_0/\omega\rceil+\omega-1}\}$. We prove that, if a scenario $\psi$
changes to $\psi'$ by adding some job $p_n$, then there exists some
job $p_n'\in R$ such that $\phi$ changes to $\phi'$ by adding
$p_n'$, and furthermore, $\phi$ and $\phi'$ are the
simulating-scenarios of $\psi$ and $\psi'$, respectively. This
suffices to approximate the above scenario sequence by the following
sequence
$$\phi_0\rightarrow\phi(1)\rightarrow \phi(2)\rightarrow\cdots \rightarrow\phi(n),$$
where $\phi(k)$ is the simulating-scenario of $\psi(k)$, and
$\phi_0$ is the initial scenario where there is no job.

We briefly argue why it is this case. Suppose $T_{LB}$ is the
scaling factor of $\psi$. According to the online algorithm, $p_n$
is put on machine $h$ where
$st_h=(\eta_{-c_0},\cdots,\eta_{\omega})$. Let
$\tau=(\nu_{-c_0},\cdots,\nu_{\omega})$ be its simulating state in
$\phi$. If $p_n/T_{LB}<(1+\epsilon)^{-c_0}$ and
$\eta_{-c_0}+p_n/T_{LB}\le \nu_{-c_0}$, then $\phi$ is still a
simulating scenario of $\psi'$ and we may set $p_n'=0$. Else if
$\nu_{-c_0}<\eta_{-c_0}+p_n/T_{LB}\le \nu_{-c_0}+1$, we may set
$p_n'=(1+\epsilon)^{-c_0}$. For the upper bound on the processing
time, suppose $p_n/T_{LB}$ is so large that the previous load of
each machine (which is no more than $4LB\le 4(1+\epsilon)^{\omega}$)
becomes no more than $(1+\epsilon)^{-c_0}p_n/T_{LB}$. It then makes
no difference by releasing an even larger job. A rigorous proof
involves a complete analysis of how the coordinates of a
trimmed-scenario change by adding a job belonging to $R$ (see
Appendix~\ref{ap:add job trimmed}), and a case by case analysis of
each possible changes between $\psi$ and $\psi'$ (see
Appendix~\ref{ap:simulate}).

\section{Constructing a Transformation Graph}
We construct a graph $G$ that contains all the possible sequences of
the form $\phi_0\rightarrow\phi(1)\rightarrow
\phi(2)\rightarrow\cdots \rightarrow\phi(n)$. This is called a
transformation graph. For ease of our following analysis, some of
the feasible trimmed-scenarios should be deleted. Recall that $1\le
\max\{1/mLD(\phi),P_{max}(\phi)\}<
(1+\epsilon)^{\omega}+2(1+\epsilon)^{-c_0}$ is satisfied for any
feasible trimmed-scenario $\phi$, and it may happen that two
trimmed-scenarios are essentially the same. Indeed, if
$(1+\epsilon)^{\omega}\le \max\{1/mLD(\phi),P_{max}(\phi)\}<
(1+\epsilon)^{\omega}+2(1+\epsilon)^{-c_0}$, then by dividing
$(1+\epsilon)^{\omega}$ from the processing times of each job in
$\phi$ we can derive another trimmed-scenario $\phi'$ satisfying
$1\le \max\{1/mLD(\phi'),P_{max}(\phi')\}<
1+2(1+\epsilon)^{-c_0-\omega}$, which is also feasible. If $\phi$ is
a simulating-scenario of $\psi$, then $\phi'$ is called a shifted
simulating-scenario of $\psi$. It is easy to verify that the instant
approximation ratio of a shifted simulating scenario is also similar
to that of the corresponding scenario (see
Appendix~\ref{ap:delete}). In this case $\phi$ is deleted and we
only keep $\phi'$. Let $\Phi'\subset\Phi$ be the set of remaining
trimmed-scenarios. We can prove that, for any real schedule
represented as $\psi(1)\rightarrow \psi(2)\rightarrow\cdots
\rightarrow\psi(n)$, we can find
$\phi_0\rightarrow\phi(1)\rightarrow \phi(2)\rightarrow\cdots
\rightarrow\phi(n)$ such that $\phi(k)\in \Phi'$ is either a
simulating-scenario or a shifted simulating-scenario of $\psi(k)$.
The reader can refer to Appendix~\ref{ap:delete} for a rigorous
proof.

Recall that when a trimmed-scenario changes to another, the
adversary only releases a job belonging to $R$. Let $\zeta=|R|$ and
$\alpha_1$, $\cdots$, $\alpha_{\zeta}$ be all the distinct
processing times in $R$. We show how $G$ is constructed.

We first construct two disjoint vertex sets $S_0$ and $A_0$. For
every $\phi_i\in \Phi'$, there is a vertex $s_i^0\in S_0$. For each
$s_i^0$, there are $\zeta$ vertices of $A_0$ incident to it, namely
$a_{ij}^0$ for $1\le j\le \zeta$. The node $a_{ij}^0$ represents the
release of a job of processing time $\alpha_j$ to the
trimmed-scenario $\phi_i$. Thus, $S_0\cup A_0$ along with the edges
forms a bipartite graph.

Let $S_1=\{s_i^1|s_i^0\in S_0\}$ be a copy of $S_0$. By scheduling a
job of $\alpha_j$, if $\phi_i$ could be changed to $\phi_k$, then
there is an edge between $a_{ij}^0$ and $s_k^1$. We go on to build
up the graph by creating an arbitrary number of copies of $S_0$ and
$A_0$, namely $S_1$, $S_2$, $\cdots$ and $A_1$, $A_2$, $\cdots$ such
that $S_h=\{s_i^h|s_i^0\in S_0\}$, $A_h=\{a_{ij}^h|a_{ij}^0\in
A_0\}$. Furthermore, there is an edge between $s_i^h$ and $a_{ij}^h$
if and only if there is an edge between $s_i^0$ and $a_{ij}^0$, and
an edge between $a_{ij}^h$ and $s_{k}^{h+1}$ if and only if there is
an edge between $a_{ij}^0$ and $s_k^1$.

The infinite graph we construct above is the transformation graph
$G$. We let $G_n$ be the subgraph of $G$ induced by the vertex set
$(\cup_{i=0}^{n}S_i)\cup(\cup_{i=0}^{n-1}A_i)$.

\section{Best Response Dynamics}
Recall that We can view online scheduling as a game between the
scheduler and the adversary. According to our previous analysis, we
can focus on trimmed-scenarios and assume that the adversary always
releases a job with processing time belonging to $R$. By scheduling
a job released by the adversary, the current trimmed-scenario
changes into another one.

We can consider the instant approximation ratio as the utility of
the adversary who tries to maximize it by leading the scheduling
into a (trimmed) scenario. After releasing $n$ jobs, if he is
satisfied with the current instant approximation ratio, then he
stops and the game is called an $n$-stage game. Otherwise he goes on
to release more jobs. The scheduler, however, tries to minimize the
competitive ratio by leading the game into trimmed-scenarios with
small instant approximation ratios.

Consider any $n$-stage game and define $\rho_n(s_k^n)=\rho(\phi_k)$.
It implies that if the game arrives at $\phi_k$ eventually, then the
utility of the adversary is $\rho(\phi_k)$. Notice that the
adversary could release a job of processing time $0$, thus $n$-stage
games include $k$-stage games for $k<n$. Consider $a_{ij}^{n-1}$. If
the current trimmed-scenario is $\phi_i$ and the adversary releases
a job with processing time $\alpha_j$, then all the possible
schedules by adding this job to different machines could be
represented by $N(a_{ij}^{n-1})=\{s_k^n:\textrm{$s_k^n$ is incident
to $a_{ij}^{n-1}$}\}$. The scheduler tries to minimize the
competitive ratio, and he knows that it is the last job, thus he
would choose the one with the least instant approximation ratio.
Thus we define
$$\rho_n(a_{ij}^{n-1})=\min_k\{\rho_n(s_k^n):s_k^n\in
N(a_{ij}^{n-1})\}.$$

Knowing this beforehand, the adversary chooses to release a job
which maximizes $\rho_n(a_{ij}^{n-1})$. Let
$N(s_i^{n-1})=\{a_{ij}^{n-1}:\textrm{$a_{ij}^{n-1}$ is incident to
$s_i^{n-1}$}\}$ and thus we define
$$\rho_n(s_i^{n-1})=\max_j\{\rho_n(a_{ij}^{n-1}):a_{ij}^{n-1}\in N(s_i^{n-1})\}.$$

Iteratively applying the above argument, we can define
$$\rho_n(a_{ij}^{h-1})=\min_k\{\rho_n(s_k^n):s_k^h\in N(a_{ij}^{h-1})\},$$
$$\rho_n(s_i^{h-1})=\max_j\{\rho_n(a_{ij}^{h-1}):a_{ij}^{h-1}\in N(s_i^{h-1})\}.$$

The value $\rho_{n}(s_i^{h})$ means that, if the current
trimmed-scenario is $\phi_i$, then the largest utility the adversary
could achieve by releasing $n-h$ jobs is $\rho_{n}(s_i^{h})$. Notice
that we start from the empty schedule $s_0^0$, thus
$\rho_n({s_0^0})$ is the largest utility the adversary could achieve
by releasing $n$ jobs.

\subsection{Bounding the number of stages}
The computation of the utility of the adversary relies on the number
of jobs released, however, theoretically the adversary could release
as many jobs as he wants. In this section, we prove the following
theorem.

\begin{theorem}
\label{th:stage bound} There exists some integer $n_0\le
O((m+1)^{\Lambda}/\epsilon)$, such that
$\rho_n(s_i^0)=\rho_{n_0}(s_i^0)$ for any $\phi_i\in \Phi'$ and
$n\ge n_0$.
\end{theorem}

To prove it, we start with the following simple lemmas.
\begin{lemma}
For any $1\le h\le n$, $\rho_{n}(s_i^h)\le \rho_n(s_i^{h-1})$.
\end{lemma}
\begin{proof}
The proof is obvious by noticing that the adversary could release a
job with processing time $0$.
\end{proof}

\begin{lemma}
For any $0\le h\le n$ and $i\neq \Gamma+1$, $\rho_{n}(s_i^h)\in
\Delta$.
\end{lemma}
\begin{proof}
The lemma clearly holds for $h=n$. Suppose the lemma holds for some
$h\ge 1$, we prove that the lemma is also true for $h-1$.

Recall that $\rho_n(a_{ij}^{h-1})=\min_k\{\rho_n(s_k^n):s_k^h\in
N(a_{ij}^{h-1})\}$. We prove that $\rho_n(a_{ij}^{h-1})\in \Delta$.
To this end, we only need to show that, we can always put $\alpha_j$
to a certain machine so that $\phi_i$ is not transformed into
$\phi_{\Gamma+1}$.

We apply list scheduling when $\alpha_j$ is released. Suppose by
scheduling $\alpha_j$ in this way, $\phi_i$ is transformed into
$\phi_{\Gamma+1}$, then $\alpha_j=(1+\epsilon)^{\mu}$ for $1\le
\mu\le \omega$ and
$LB'=\max\{1/mLD(\phi)+\alpha_j/m,P_{max}(\phi),\alpha_j\}<(1+\epsilon)^{\omega}+2(1+\epsilon)^{-c_0}$.
Furthermore, suppose $\alpha_j$ is put to a machine whose
trimmed-state is $\tau$. Then $LD(\tau)+\alpha_j\ge
4(1+\epsilon)^{\omega}+2(1+\epsilon)^{-c_0}$. Now it follows
directly that $LD(\tau)> 3(1+\epsilon)^{\omega}$. Notice that we put
$\alpha_j$ to the machine with the least load. Before $\alpha_j$ is
released, the load of every machine in $\phi_i$ is larger than
$3(1+\epsilon)^{\omega}$, which contradicts the fact that $\phi_i$
is a feasible trimmed-scenario.

Therefore, applying list scheduling, $\phi_i$ can always transform
to another feasible trimmed-scenario, which ensures that
$\rho_n(a_{ij}^{h-1})\in \Delta$. Thus
$\rho_n(s_i^{h-1})=\max_j\{\rho_n(a_{ij}^{h-1}):a_{ij}^{h-1}\in
N(s_i^{h-1})\}\in \Delta$.

\end{proof}

\begin{lemma}
If there exists a number $n\in N$ such that $\rho_{n+1}(s_i^0)=
\rho_{n}(s_i^{0})$, then for any integer $h\ge 0$,
$\rho_{n+h}(s_i^0)=\rho_{n}(s_i^{0})$.
\end{lemma}
\begin{proof}
We prove the lemma by induction. Suppose it holds for $h$. We
consider $h+1$.

Obviously
$\rho_{n+h}(s_i^{n+h})=\rho_{n+h+1}(s_i^{n+h+1})=\rho(\phi_i)$.
According to the computing rule,
$$\rho_{n+h+1}(a_{ij}^{n+h})=\min_k\{\rho_{n+h+1}(s_k^{n+h+1}):s_k^{n+h+1}\in N(a_{ij}^{n+h})\},$$
$$\rho_{n+h}(a_{ij}^{n+h-1})=\min_k\{\rho_{n+h}(s_k^{n+h}):s_k^{n+h}\in N(a_{ij}^{n+h-1})\}.$$
Recall that $s_k^{n+h+1}\in N(a_{ij}^{n+h})$ if and only if
$s_k^1\in N(a_{ij}^0)$, and thus it is also equivalent to
$s_k^{n+h}\in N(a_{ij}^{n+h-1})$. Hence,
$\rho_{n+h+1}(a_{ij}^{n+h})=\rho_{n+h}(a_{ij}^{n+h-1})$.

Using analogous arguments, we can show that
$\rho_{n+h+1}(s_{i}^{n+h})=\rho_{n+h}(s_{i}^{n+h-1})$. Iteratively
applying the above procedure, we can finally show that
$\rho_{n+h+1}(s_{i}^{1})=\rho_{n+h}(s_{i}^{0})$. Similarly,
$\rho_{n+h}(s_{i}^{1})=\rho_{n+h-1}(s_{i}^{0})$.

According to the induction hypothesis, we know
$\rho_{n+h}(s_{i}^{1})=\rho_{n+h-1}(s_{i}^{0})=\rho_{n}(s_i^{0})$,
and
$\rho_{n+h+1}(s_{i}^{1})=\rho_{n+h}(s_{i}^{0})=\rho_{n}(s_i^{0})$.
Meanwhile
$$\rho_{n+h}(a_{ij}^{0})=\min_k\{\rho_{n+h}(s_k^{1}):s_k^{1}\in N(a_{ij}^{0})\}=\min_k\{\rho_{n}(s_k^{0}):s_k^{1}\in N(a_{ij}^{0})\},$$
$$\rho_{n+h+1}(a_{ij}^{0})=\min_k\{\rho_{n+h+1}(s_k^{1}):s_k^{1}\in N(a_{ij}^{0})\}=\min_k\{\rho_{n}(s_k^{0}):s_k^{1}\in N(a_{ij}^{0})\}.$$
Thus it immediately follows that
$\rho_{n+h}(a_{ij}^{0})=\rho_{n+h+1}(a_{ij}^{0})$. Furthermore,
\begin{eqnarray*}
\rho_{n+h+1}(s_i^{0})&=&\max_j\{\rho_{n+h+1}(a_{ij}^{0}):a_{ij}^{0}\in
N(s_i^{0})\}\\
&=&\max_j\{\rho_{n+h}(a_{ij}^{0}):a_{ij}^{0}\in
N(s_i^{0})\}=\rho_{n+h}(s_i^{0}). \end{eqnarray*} The lemma holds
for $h+1$.
\end{proof}

Now we arrive at the proof of Theorem \ref{th:stage bound}. Define
$Z(n)=\sum_{\phi_i\in\Phi'\setminus\{\phi_{\Gamma+1}\}}\rho_n(s_i^0)$
as the potential function. According to the previous lemmas,
$Z(n+1)\ge Z(n)$, and if $Z(n_0+1)=Z(n_0)$, then $Z(n)=Z(n_0)$ for
any $n\ge n_0$. Furthermore, if $Z(n+1)>Z(n)$, then $Z(n+1)-Z(n)\ge
\epsilon$. Suppose $Z(n+1)>Z(n)$, then it follows directly that
$Z(n+1)>Z(n)>\cdots>Z(1)$. Recall that $Z(1)\ge 0$ and $Z(n+1)\le
20(|\Phi'|-1)\le O((m+1)^{\Lambda})$, thus $n+1\le
O((m+1)^{\Lambda}/\epsilon)$. Furthermore, it can be easily verified
that if $Z(n+1)=Z(n)$, then $\rho_{n+1}(s_i^0)=\rho_n(s_i^0)$ for
any $\phi\in\Phi'$. Thus, by setting
$n_0=O((m+1)^{\Lambda}/\epsilon)$, Theorem \ref{th:stage bound}
follows.

Let $n_0$ be the smallest integer satisfying Theorem \ref{th:stage
bound}. Let $\rho^*=\rho_{n_0}(s_0^0)$, and
$\rho(s_i^0)=\rho_{n_0}(s_i^0)$. Now it is not difficult to see
that, the optimal online algorithm for $P||C_{max}$ has a
competitive ratio around $\rho^*$. A rigorous proof of such an
observation depends on the following two facts.
\begin{itemize}
\item[1. ]{Given any online algorithm, there exists a list
of at most $n_0$ jobs such that by scheduling them, its competitive
ratio exceeds $\rho^*-O(\epsilon)$.}
\item[2. ]{There exists an online algorithm whose competitive ratio is at most $\rho^*+O(\epsilon)$.}
\end{itemize}
The first fact could be proved via $G_{n_0}$, where
$\rho^*=\rho_{n_0}(s_0^0)$ ensures that $n_0$ jobs are enough to
achieve the lower bound. The readers may refer to Appendix
\ref{ap:strategy adv} for details. The second observation could be
proved via $G_{n_0+1}$, where
$\rho_{n_0+1}(s_i^0)=\rho_{n_0+1}(s_i^1)=\rho(s_i^0)$ for every
$\phi_i$. Each time a job is released, the scheduler may assume that
he is at the vertex $s_i^0$ where $\rho_{n_0+1}(s_i^0)\le \rho^*$,
and find a feasible schedule by leading the game into $s_k^1$ where
$\rho_{n_0+1}(s_i^0)=\rho_{n_0+1}(s_k^1)\le \rho^*$. After
scheduling the job he may still assume that he is at $s_k^0$. The
readers may refer to Appendix~\ref{ap:strategy sche} for details.

Using the framework we derive, competitive schemes could be
constructed for a variety of online scheduling problems, including
$Rm||C_{max}$ and $Rm||\sum_i C_i^p$ for constant $p$. Additionally,
if we restrict that the processing time of each job is bounded by
$q$, then an optimal online algorithm for $P|p_j\le q|C_{max}$ could
be derived (in $(mq)^{O(mq)}$ time). The readers may refer to
Appendix~\ref{ap:extension} for details.

\section{Concluding Remarks}
We provide a new framework for the online over list scheduling
problems. We remark that, through such a framework, nearly optimal
algorithms could also be derived for other online problems,
including the k-server problem (despite that the running time is
rather huge, which is exponential).

As nearly optimal algorithms could be derived for various online
problems, it becomes a very interesting and challenging problem to
consider the hardness of deriving optimal online algorithms. Is
there some complexity domain such that finding an optimal online
algorithm is {\em hard} in some sense? For example, given a constant
$\rho$, consider the problem of determining whether there exists an
online algorithm for $P||C_{max}$ whose competitive ratio is at most
$\rho$. Could it be answered in time $f(m,\rho)$ for any given
function $f$? We expect the first exciting results along this line,
that would open the online area at a new stage.

\bibliographystyle{plain}

\clearpage

\begin{appendix}

\section{Adding a new job to a scenario}\label{ap:add job scenario} Before we show how a scenario is changed
by adding a new job, we first show how a scenario is changed when we
scale its jobs using a new factor $T\in SC$ and $T>T_{LB}$.

\subsection{Re-computation of a scenario}
Let $(T_{LB},\psi)$ be a real schedule at any time where
$\psi=(st_1,st_2,\cdots,st_m)$. If we choose $T>T_{LB}$ to scale
jobs, then a big job previously may become a small job (i.e., no
greater than $T(1+\epsilon)^{-c_0}$). Suppose
$T=T_{LB}(1+\epsilon)^{k\omega}$, then a job with processing time
$T_{LB}(1+\epsilon)^{j}$ is denoted as $T(1+\epsilon)^{j-k\omega}$
now, hence a state $st=(\eta_{-c_0},\cdots,\eta_{\omega})$ of $\psi$
becomes $\hat{st}=(\hat{\eta}_{-c_0},\cdots,\hat{\eta}_{\omega})$
where $\hat{\eta}_i=\eta_{i+k\omega}$ for $i>-c_0$ (we let
$\eta_{i}=0$ for $i>\omega$), and
$$\hat{\eta}_{c_0}=\frac{\sum_{i=-c_0}^{k\omega-c_0}T_{LB}(1+\epsilon)^i\eta_i}{T(1+\epsilon)^{-c_0}}
=\frac{\sum_{i=-c_0}^{k\omega-c_0}(1+\epsilon)^i\eta_i}{(1+\epsilon)^{k\omega-c_0}}.$$

The above computation could be viewed as shifting the state
leftwards by $k\omega$ 'bits', and we define a function $f_k$ to
represent it such that $f_k(st)=\hat{st}$. Similarly the scenario
$\psi$ changes to $\hat{\psi}=(f_k(st_1,\cdots,f_k(st_m))$ and we
denote $f_k(\psi)=\hat{\psi}$.

\subsection{Adding a new job}
Again, let $(T_{LB},\psi)$ be a real schedule at any time where
$\psi=(st_1,st_2,\cdots,st_m)$. Suppose a new job $p_n$ is released
and scheduled on machine $h$ where
$st_h=(\eta_{-c_0},\eta_{-c_0+1},\cdots,\eta_{\omega})$, and
furthermore, $\psi$ changes to $\psi'$. We determine the coordinates
of $\psi'$ in the following.

Consider $p_n$. If $p_n\le T_{LB}(1+\epsilon)^{\omega}$ then we
define the addition $st_h+p_n/T_{LB}=\bar{st}_h$ in the following
way where
$\bar{st_h}=(\bar{\eta}_{-c_0},\cdots,\bar{\eta}_{\omega})$.
\begin{itemize}
\item{If $p_n/T=(1+\epsilon)^{\mu}$ for $-c_0+1\le \mu\le \omega$, then $\bar{\eta}_{\mu}=\eta_{\mu}+1$ and $\bar{\eta_j}=\eta_j$ for $j\neq \mu$.}
\item{If $p_n/T\le (1+\epsilon)^{-c_0}$, then $\bar{\eta}_{-c_0}=\eta_{-c_0}+p_n/(T_{LB}(1+\epsilon)^{-c_0})$ and $\bar{\eta}_j=\eta_j$ for $j\neq -c_0$.}
\end{itemize}

Let
$\bar{\psi}=(st_1,\cdots,st_{h-1},\bar{st}_h,st_{h+1},\cdots,st_m)$
be a temporal result. If $\bar{\psi}$ is feasible, which implies
that $\max\{LD(\bar{\psi})/m,P_{max}(\bar{\psi})\}\in
[1,(1+\epsilon)^{\omega})$, then $\psi'=\bar{\psi}$. Otherwise
$\bar{\psi}$ is infeasible and there are two possibilities.

\noindent\textbf{Case 1. }$\max\{1/mLD(\bar{\psi}),
P_{max}(\bar{\psi})\}\ge(1+\epsilon)^{\omega}$. It is not difficult to verify that\\
$\max\{1/mLD(\bar{\psi}),P_{max}(\bar{\psi})\}<(1+\epsilon)^{2\omega}$,
thus $f_1(\bar{\psi})$ is feasible and we write
$\psi'=f_1(\bar{\psi})$.

\noindent\textbf{Case 2. }$1\le \max\{1/mLD(\bar{\psi}),
P_{max}(\bar{\psi})\}<(1+\epsilon)^{\omega}$ while
$LD(\bar{st}_h)>4(1+\epsilon)^{\omega}$, i.e., $\bar{st}_h$ is an
infeasible state. In this case the competitive ratio of the online
algorithm becomes larger than $2$. Thus job $p_n$ is never added to
$st_h$ if it is scheduled according to an online algorithm with
competitive ratio no greater than $2$.

Otherwise, $(1+\epsilon)^{k\omega}\le p_n/T_{LB}<
(1+\epsilon)^{(k+1)\omega}$ for some $k\ge1$. It is easy to verify
that, by adding $p_n$ to the schedule, the scaling factor becomes
$T_{LB}(1+\epsilon)^{k\omega}$. Thus $\psi'=(st_1',\cdots,st_m')$
where $st_j'=f_k(st_j)$ for $j\neq h$, and
$st_h'=f_k(st_h)+p_n/(T_{LB}(1+\epsilon)^{\omega})$.

\section{Adding a new job to a trimmed-scenario}
\label{ap:add job trimmed} Notice that a trimmed-scenario could also
be viewed as a scenario, thus adding a new job to it could be viewed
as adding a new job to a scenario, and then rounding up the
coordinates of the resulted scenario to integers. Specifically, we
restrict the processing time of the job added is either $0$ or
$(1+\epsilon)^{\mu}$ for $\mu\ge -c_0$. We will show later that it
is possible to put an upper bound on the processing times.

\subsection{Re-computation of a trimmed-scenario}
To re-compute a trimmed-scenario $\phi$, we take $\phi$ as a
scenario with scaling factor $T_{LB}=1$. Suppose we want to use a
new factor $(1+\epsilon)^{\omega}$ to scale jobs, then each
trimmed-state of $\phi$, say $\tau$, is re-computed as $f_1(\tau)$.
Notice that its first coordinate may be fractional, we round it up
and let $g_1(\tau)=\lceil f_1(\tau)\rceil$ where $\lceil
\vec{v}\rceil$ for a vector means we round each coordinate $v_i$ of
$\vec{v}$ to $\lceil v_i\rceil$.

We define $g_k$ iteratively as $g_k(\tau)=g_{k-1}(g_1(\tau))$.

Notice that if $\tau$ is feasible (i.e., $LD(\tau)\le
4(1+\epsilon)^{\omega}+2(1+\epsilon)^{-c_0}$), then $g_k(\tau)$ is
feasible for any $k\ge 1$. Thus, we define
$g_k(\phi)=\phi'=(\xi_1',\xi_2',\cdots,\xi_{\Lambda}')$ where
$\xi_j'=\sum_{h:g_k(\tau_h)=\tau_j}\xi_h$. Specifically, if
$\{h:g_k(\tau_h)=\tau_j\}=\emptyset$, then $\xi_j'=0$.

We have the following lemma.
\begin{lemma}
\label{le:state trans} For any integer $k\ge 0$, feasible state
$st_h$ and feasible trimmed-state $\tau$, the following holds:
$$(1+\epsilon)^{k\omega}LD(f_k(st_h))=LD(st_h),$$
$$LD(\tau)\le (1+\epsilon)^{k\omega}LD(g_k(\tau))\le LD(\tau)+\sum_{i=1}^k(1+\epsilon)^{-c_0+i\omega}\le LD(\tau)+2(1+\epsilon)^{-c_0+k\omega}.$$
\end{lemma}
The proof is simple through induction.

\subsection{Adding a new job}
Suppose the feasible trimmed-scenario $\phi$ becomes $\phi'$ by
adding a new job $p_n=(1+\epsilon)^{\mu}$, and furthermore, the job
is added to a machine whose trimmed-state is $\tau_j$. We show how
the coordinates of $\phi'$ is determined.

There are two possibilities.

\noindent\textbf{Case 1. }If $-c_0\le \mu\le \omega$, then by adding
a new job $p_n=(1+\epsilon)^{\mu}$ to a feasible trimmed-state
$\tau_j$, we simply take $\tau_j$ as a state and compute
$\bar{\tau}_j=\tau_j+p_n$ according to the rule of adding a job to
states.

Consider the $m$ trimmed-states of $\phi$, we replace $\tau_j$ with
$\bar{\tau}_j$ while keeping others intact. By doing so a temporal
trimmed-scenario $\bar{\phi}$ is generated and we compute
$LB(\bar{\phi})=\max\{1/mLD(\phi)+p_n/m,P_{max}(\phi),p_n\}$. There
are three possibilities.

\noindent\textbf{Case 1.1 }$LB(\bar{\phi})<
(1+\epsilon)^{\omega}+2(1+\epsilon)^{-c_0}$ and
$LD(\bar{\tau}_j)<4(1+\epsilon)^{\omega}+2(1+\epsilon)^{-c_0}$. Then
$\bar{\tau}_j$ is a feasible trimmed-state and suppose
$\bar{\tau}_j=\tau_{j'}$. Then $\phi'=\bar{\phi}$, i.e.,
$\phi'=(\xi_1',\xi_2',\cdots,\xi_{\Lambda}')$ where
$\xi_j'=\xi_j-1$, $\xi_{j'}'=\xi_{j'}+1$ and $\xi_l'=\xi_l$ for
$l\neq j,j'$.

\noindent\textbf{Case 1.2 }$LB(\bar{\phi})< (1+\epsilon)^{\omega}$
and
$LD(\bar{\tau}_j)\ge4(1+\epsilon)^{\omega}+2(1+\epsilon)^{-c_0}$.
Then $\bar{\tau}_j$ is infeasible and $\phi'=\phi_{\Gamma+1}$.

\noindent\textbf{Case 1.3 }$LB(\bar{\phi})\ge
(1+\epsilon)^{\omega}$. It can be easily verified that
$LB(\bar{\phi})<(1+\epsilon)^{2\omega}$. Notice that
$g_1(\bar{\tau}_j)$ is always feasible, thus
$\phi'=g_1(\bar{\phi})$, i.e., for each trimmed-state $\tau$ of
$\bar{\phi}$, we compute $g_1(\tau)$. Since $g_1(\tau)$ is always
feasible, they made up of a feasible trimmed-scenario $\phi'$.

\noindent\textbf{Remark. }There might be intersection between Case 1
and Case 3. Indeed, if $(1+\epsilon)^{\omega}\le
LB(\bar{\phi})<(1+\epsilon)^{\omega}+2(1+\epsilon)^{-c_0}$, and
$\bar{\tau}$ is feasible, then by adding $p_n$ the trimmed-scenario
$\phi$ changes into $\bar{\phi}=\phi'$ according to Case 1 and
$g_1(\phi')$ according to Case 3. Here both $\phi'$ and $g_1(\phi')$
are feasible trimmed-scenarios.

This is the only case that $\phi+p_n$ may yield two different
solutions. In the next section we will remove $\phi$ if both $\phi$
and $g_1(\phi)$ are feasible. By doing so $\phi+p_n$ yields a unique
solution, but currently we just keep both of them so that
Theorem~\ref{th:simulate} could be proved.

\noindent\textbf{Case 2. }If $(1+\epsilon)^{k\omega}\le
\mu<(1+\epsilon)^{(k+1)\omega}$ then again we take $\tau_j$ as a
state and compute
$\bar{\tau}_j=g_k(\tau_j)+p_n/(1+\epsilon)^{k\omega}$.

We re-compute $\phi$ as
$g_{k}(\phi)=(\hat{\xi}_1,\hat{\xi}_2,\cdots,\hat{\xi}_{\Lambda})$.
Then we replace one trimmed-state $g_k(\tau_j)$ with $\bar{\tau}_j$
and this generates $\phi'$. It is easy to verify that $\phi'$ is
feasible.

\noindent\textbf{Remark 2. }Notice that the number of possible
processing times of job $p_n$ could be infinite, however, we show
that it is possible to further restrict it to be some constant.

Let $p_n=(1+\epsilon)^{\mu}$. Let $\mu_0$ be the smallest integer
such that $(1+\epsilon)^{\mu_0}\ge 4(1+\epsilon)^{\omega+c_0+1}$. If
$\mu=k\omega+l$ with $k\ge \lceil \mu_0/\omega\rceil$ and $0\le l\le
\omega-1$, then $\phi$ is re-computed as $g_{k}(\phi)$. Notice that
for any feasible trimmed-state $\tau$, $LD(\tau)\le
4(1+\epsilon)^{\omega}+2(1+\epsilon)^{-c_0}<4(1+\epsilon)^{\omega+1}$,
thus $LD(g_k(\tau))\le (1+\epsilon)^{-c_0}$, which implies that
$g_k(\tau)=(0,0,\cdots,0)$ if $\tau=(0,0,\cdots,0)$ and
$g_k(\tau)=(1,0,0,\cdots,0)$ otherwise. Thus, $g_k(\phi)=g_{\lceil
\mu_0/\omega\rceil}(\phi)$.

The above analysis shows that by adding a job with processing time
$p_n=(1+\epsilon)^{k\omega+l}$ for $k\ge \lceil \mu_0/\omega\rceil$
and $0\le l\le \omega-1$ to any feasible trimmed-scenario $\phi$ is
equivalent to adding a job with processing time
$p_n=(1+\epsilon)^{\lceil \mu_0/\omega\rceil\omega+l}$ to $\phi$.

Thus, when adding a job to a trimmed-scenario, we may restrict that
$p_n\in R=\{0,(1+\epsilon)^{-c_0},\cdots,(1+\epsilon)^{\lceil
\mu_0/\omega\rceil+\omega-1}\}$.

\section{Simulating transformations between scenarios}
\label{ap:simulate} The whole section is devoted to prove the
following theorem.

\begin{theorem}
\label{th:simulate} Let $\phi$ be the simulating-scenario of a
feasible scenario $\psi$. If according to some online algorithm
$(T,\psi)$ changes to $(T',\bar{\psi})$ by adding a job $p_n\neq 0$,
then $\phi$ could be transformed to $\bar{\phi}$ ($\bar{\phi}\neq
\phi_0,\phi_{\Gamma}$) by adding a job $p_n'\in
R=\{0,(1+\epsilon)^{-c_0},\cdots,(1+\epsilon)^{\lceil
\mu_0/\omega\rceil+\omega-1}\}$ such that $\bar{\phi}$ is a
simulating-scenario of $\bar{\psi}$.
\end{theorem}
Let $\tau_{\theta(h)}$ in $\phi$ be the simulating-state of $st_h$
in $\psi$. Before we give the proof, we first present a lemma that
would be used later.
\begin{lemma}
\label{le:simu-g_k} Let $\phi$ be a simulating-scenario of $\psi$.
For any $k\ge 1$, if
$f_k(st_h)=(\eta_{-c_0}',\eta_{-c_0+1}',\cdots,\eta_{\omega}')$ and
$g_k(\tau_{\theta(h)})=(\nu_{-c_0}',\nu_{-c_0+1}',\cdots,\nu_{\omega}')$,
then $\nu_{i}'=\eta_i'$ for $i>-c_0$ and $\eta_{-c_0}'\le
\nu_{-c_0}'\le \eta_{-c_0}'+2$.
\end{lemma}
\begin{proof}
Let $st_h=(\eta_{-c_0},\eta_{-c_0+1},\cdots,\eta_{\omega})$ and
$\tau_{\theta(h)}=(\nu_{-c_0},\nu_{-c_0+1},\cdots,\nu_{\omega})$. We
first prove the lemma for $k=1$.

It is easy to verify that $\nu_i'=\eta_i'$ for $i>-c_0$.
Furthermore,
\begin{eqnarray*}
{\nu}_{-c_0}'&=&
\lceil \frac{\sum_{i=-c_0}^{\omega-c_0}(1+\epsilon)^i\nu_i}{(1+\epsilon)^{\omega-c_0}}\rceil\\
&\le& \frac{\sum_{i=-c_0+1}^{\omega-c_0}(1+\epsilon)^i\eta_i+(1+\epsilon)^{-c_0}(\eta_{-c_0}+2)}{(1+\epsilon)^{\omega-c_0}}+1\\
&\le & {\eta}_{-c_0}'+1+2(1+\epsilon)^{-\omega}< {\eta}_{-c_0}'+2
\end{eqnarray*}
Thus the lemma holds for $k=1$.

If the lemma holds for $k=k_0$, then it also holds for $k=k_0+1$.
The proof is the same.
\end{proof}
Now we come to the proof of Theorem \ref{th:simulate}.
\begin{proof}
Let $\psi=(st_1,st_2,\cdots,st_m)$ and
$\phi=(\xi_1,\xi_2,\cdots,\xi_{\Lambda})$. Recall that
$\tau_{\theta(i)}$ is the simulating-state of $st_i$ in $\phi$.

Notice that $LD(st_i)\le LD(\tau_{\theta(i)})\le
LD(st_i)+2(1+\epsilon)^{-c_0}$, it follows that $1/mLD(\psi)\le
1/mLD(\phi)\le 1/mLD(\psi)+2(1+\epsilon)^{-c_0}$. Meanwhile,
$P_{max}(\psi)=P_{max}(\phi)$ as long as $\psi\neq (0,0,\cdots,0)$.

Suppose job $n$ is assigned to machine $h$ in the real schedule. Let
$st_h=(\eta_{-c_0},\cdots,\eta_{\omega})$ and
$\tau_{\theta(h)}=(\nu_{-c_0},\cdots,\nu_{\omega})$. Recall that
$\eta_{-c_0}\le \nu_{c_0}\le \eta_{-c_0}+2$ and $\eta_{i}=\nu_i$ for
$i>c_0$.

There are two possibilities.

\noindent\textbf{Case 1. } $$p_n/T\le (1+\epsilon)^{\omega}.$$ Let
$st_h'=st_h+p_n/T=(\eta_{-c_0}',\cdots,\eta_{\omega}')$. We define
$p_n'$ in the following way.
\begin{itemize}
\item{If $p_n/T=(1+\epsilon)^{\mu}$ for $-c_0+1\le \mu\le\omega$, then $p_n'=p_n/T$.}
\item{If $p_n/T\le (1+\epsilon)^{-c_0}$,}
\begin{itemize}
\item{$\eta_{-c_0}'\le \nu_{-c_0}$, then $p_n'=0$.}
\item{$\eta_{-c_0}'> \nu_{-c_0}$, then $p_n'=(1+\epsilon)^{-c_0}$.}
\end{itemize}
\end{itemize}
Let $\tau_{\theta(h)}'+p_n'=(\nu_{-c_0}',\cdots,\nu_{\omega}')$,
then $\nu_{-c_0}'=\nu_{-c_0}+p_n'/((1+\epsilon)^{-c_0})$, in both
cases $\eta_{-c_0}'\le \nu_{-c_0}'\le \eta_{-c_0}'+2$.

By adding $p_n$ to $\psi$, the scaling factor may or may not be
changed.

If $T=T'$, the state of machine $h$ in $\bar{\psi}$ is $st_h'$. We
consider $\tau_{\zeta(h)}+p_n'$. Since
$LD(\tau_{\zeta(h)}+p_n')-LD(st_h')\le
LD(\tau_{\zeta(h)})-LD(st_h)\le 2(1+\epsilon)^{-c_0}$, and $st_h'$
is a feasible state, $\tau_{\zeta(h)}+p_n'$ is also a feasible
trimmed state. Meanwhile
$\max\{1/mLD(\psi'),P_{max}(\psi')\}<(1+\epsilon)^{\omega}$, thus by
adding $p_n'$ to $\phi$, the scaling factor of the trimmed-scenario
is also not updated, which implies that the trimmed-state of machine
$h$ in $\bar{\phi}$ is $\tau_{\zeta(h)}+p_n'$. It can be easily
verified that in this case, $\bar{\phi}$ is the simulating-scenario
of $\bar{\psi}$.

Otherwise $T'>T$ and the state of machine $h$ is $f_1(st_h')$ in
$\bar{\psi}$. We compute
$LB'=\max\{1/mLD(\phi)+p_n'/m,P_{max}(\phi),p_n'\}$. Since
$LB=\max\{1/mLD(\psi)+p_n/m,P_{max}(\psi),p_n\}>(1+\omega)^{\omega}$,
it follows directly that $LB'>(1+\omega)^{\omega}$. Meanwhile
$LB'<(1+\omega)^{2\omega}$, thus the trimmed-state of machine $h$ in
$\bar{\phi}$ is $g_1(\tau_{\zeta(h)}'+p_n')$.

We compare $st_h'$ and
$\tau_{\zeta(h)}'+p_n'=(\nu_{-c_0}',\cdots,\nu_{\omega}')$.
Obviously $\eta_{l}'=\nu_{l}'$ for $l\neq -c_0$ and $\eta_{c_0'}\le
\nu_{-c_0}'\le \eta_{-c_0}'+2$. According to Lemma
\ref{le:simu-g_k}, $g_1(\tau_{\zeta(h)}'+p_n')$ is a
simulating-state of $f_1(st_h')$, which implies that $\bar{\phi}$ is
a simulating-scenario of $\bar{\psi}$.

\noindent\textbf{Remark. }Recall that when $(1+\epsilon)^{\omega}\le
LB'<(1+\epsilon)^{\omega}+2(1+\epsilon)^{-c_0}$, $\phi+p_n'$ may
yield two solutions $\hat{\phi}$ and $g_{1}(\hat{\phi})$, as we have
claimed. Our above discussion chooses $\hat{\phi}$ if the scaling
factor of the real schedule does not change, and chooses
$g_1(\hat{\phi})$ when the the scaling factor of the real schedule
changes.

\noindent\textbf{Case 2. } For some $k\ge 1$,
$$(1+\epsilon)^{k\omega}\le p_n/T<(1+\epsilon)^{(k+1)\omega}.$$
Then we define $p_n'=p_n/T$ at first.

Let $f_{k}(st_h)=({\eta}_{-c_0}',\cdots,{\eta}_{\omega}')$,
$g_{k}(\tau_{\zeta(h)})=({\nu}_{-c_0}',\cdots,{\nu}_{\omega}')$,
then according to Lemma \ref{le:simu-g_k} we have
${\eta}_i'={\nu}_i'$ for $-c_0<i\le \omega$ and $\eta_{-c_0}'\le
\nu_{-c_0}'\le \eta_{-c_0}'+2$. Then it follows directly that
$g_{k}(\tau_{\zeta(h)})+p_n'$ is a simulating-state of
$f_{k}(st_h)+p_n$. Thus, by adding $p_n'$, $\bar{\phi}$ is a
simulating-scenario of $\bar{\psi}$.

Furthermore, if $p_n'>(1+\epsilon)^{\lceil
\mu_0/\omega\rceil+\omega-1}$, then suppose
$p_n'=(1+\epsilon)^{k'\omega+l}$ for some $k'\ge \lceil
\mu_0/\omega\rceil$ and $0\le l\le \omega-1$. Due to our previous
analysis, $p_n'$ could be replaced by a job with processing time
$p_n''=(1+\epsilon)^{\lceil \mu_0/\omega\rceil+l}$. The
trimmed-scenario $\phi$ still transforms into $\bar{\phi}$ by adding
$p_n''$.
\end{proof}

\section{Deletion of equivalent trimmed-scenarios}
\label{ap:delete} Recall that the addition $\phi+p_n$ may yield two
solutions, $\phi'$ and $g_1({\phi'})$ where both of them are
feasible. To make the result unique, $\phi'$ is deleted from $\Phi$
if $g_1(\phi')$ is feasible and $\Phi'$ is the set of the remaining
trimmed-scenarios.

We have the following simple lemma.

\begin{lemma}
If $\phi$ and $g_1(\phi)$ are both feasible trimmed-scenarios, then
$|\rho(\phi)-\rho(g_1(\phi))|\le O(\epsilon)$.
\end{lemma}

With fewer trimmed-scenarios, Theorem~\ref{th:simulate} may not
hold, however, we have the following lemma.

\begin{lemma}
\label{le:delete trimmed-scenario} Suppose by releasing job $n$ with
$p_n\in R$ and scheduling it onto a certain machine, the feasible
trimmed-scenario $\phi$ changes to $\hat{\phi}$. Furthermore,
$g_{1}(\phi)$ is also feasible. Then there exists $p_n'\in R$ such
that by scheduling it on the same machine, $g_{1}(\phi)$ changes to
$\bar{\phi}$ and furthermore, either $\bar{\phi}=\hat{\phi}$ or
$\bar{\phi}=g_1(\hat{\phi})$.
\end{lemma}
\begin{proof}
Suppose job $n$ is scheduled onto a machine of trimmed-state
$\tau=(\nu_{-c_0},\cdots,\nu_{\omega})$ in $\phi$, then we put
$p_n'$ onto a machine of trimmed-state
$g_1(\tau)=(\nu_{-c_0}',\cdots,\nu_{\omega}')$ in $g_1(\phi)$. If
$p_n=0$ then obviously we can choose $p_n'=0$. Otherwise let
$p_n=(1+\epsilon)^{\mu}$ and there are three possibilities.

\noindent\textbf{Case 1. }$\mu\le\omega-c_0$.

If by adding $p_n$, the scaling factor of $\phi$ does not change,
then we compare $\nu_{-c_0}'=\lceil
\frac{\sum_{i=-c_0}^{k\omega-c_0}(1+\epsilon)^i\nu_i}{(1+\epsilon)^{\omega-c_0}}\rceil$
with $y=\lceil
\frac{\sum_{i=-c_0}^{k\omega-c_0}(1+\epsilon)^i\nu_i+(1+\epsilon)^{\mu}}{(1+\epsilon)^{\omega-c_0}}\rceil\le
\nu_{-c_0}'+1$. If $\nu_{-c_0}'=y$, then $p_n'=0$. Otherwise
$y=\nu_{-c_0}'+1$, then $p_n'=(1+\epsilon)^{-c_0}$. It can be easily
verified that $g_1({\tau})+p_n'=g_1(\tau+p_n)$ and
$g_1(\hat{\phi})=g_1({\phi})+p_n'$.

Otherwise by adding $p_n$ the scaling factor of $\phi$ increases,
then we define $p_n'$ in the same way and it can be easily verified
that $\hat{\phi}=g_1(\phi)+p_n'$.

\noindent\textbf{Case 2. }$\omega-c_0<\mu\le 2\omega$.

In this case we define $p_n'=(1+\omega)^{\mu-\omega}$ and the proof
is similar to the previous case.

Notice that in both case 1 and case 2, $p_n'\le
(1+\omega)^{\omega}$. As $LD(g_1(\tau))\le
4+2(1+\epsilon)^{-c_0-\omega}$, $LD(g_1(\tau))+p_n'\le
4(1+\epsilon)^{\omega}$, thus we can add $p_n'$ to $g_{1}(\tau)$
directly (without changing the scaling factor). Furthermore,
$\max\{1/m[LD(g_1({\phi}))+p_n'],P_{max}(g_1(g_1(\phi))),p_n'\}\le
(1+\epsilon)^{\omega}$, thus by adding $p_n'$ to $g_1(\phi)$, the
scaling factor does not change, thus in both cases,
$\bar{\phi}=g_1(\phi)+p_n'$.

\noindent\textbf{Case 3. }$\mu>2\omega$.

Suppose $\mu=k\omega+l$ with $k\ge 2$ and $0\le l\le \omega-1$. Then
$p_n'=(1+\epsilon)^{\mu-\omega}$. According to the definition of
$g_k$, $g_k(\phi)=g_{k-1}(g_1(\phi))$, thus $\bar{\phi}=\hat{\phi}$.

\end{proof}

Combining Theorem~\ref{th:simulate} and Lemma~\ref{le:delete
trimmed-scenario}, we have the following theorem.

\begin{theorem}
\label{th:simulate strong} Let $\phi\in \Phi'$ be the
simulating-scenario or shifted simulating-scenario of a feasible
scenario $\psi$. If according to some online algorithm $(T,\psi)$
changes to $(T',\bar{\psi})$ by adding a job $p_n\neq 0$, then
$\phi$ could be transformed to $\bar{\phi}\in \Phi'$
($\bar{\phi}\neq \phi_0,\phi_{\Gamma}$) by adding a job $p_n'\in
R=\{0,(1+\epsilon)^{-c_0},\cdots,(1+\epsilon)^{\lceil
\mu_0/\omega\rceil+\omega-1}\}$ such that $\bar{\phi}$ is a
simulating-scenario or shifted simulating-scenario of $\bar{\psi}$.
\end{theorem}

\section{The nearly optimal strategies for the adversary and the scheduler}
\subsection{The nearly optimal strategy for the adversary}
\label{ap:strategy adv} We prove in this subsection that, by
releasing at most $n_0$ jobs, the adversary can ensure that there is
no online algorithm whose competitive ratio is less than
$\rho^*-O(\epsilon)$.

We play the part of the adversary.

Consider $G_{n_0}$. Notice that
$\rho^*=\rho_{n_0}(s_0^0)=\max_j\{\rho_{n_0}(a_{0,j}^{0}):a_{0,j}^{0}\in
N(s_0^{0})\}$, thus there exists some $j_0$ such that
$a_{0,j_0}^{0}\in N(s_0^{0})$ and
$\rho_{n_0}(a_{0,j_0}^{0})=\rho^*$.

We release a job with processing time $\alpha_{j_0}$. Suppose due to
any online algorithm whose competitive ratio is no greater than $2$,
this job is scheduled onto a certain machine so that the scenario
becomes $\psi$, then according to Theorem \ref{th:simulate} and the
construction of the graph, there exists some $s_k^1$ incident to
$a_{0,j_0}^{0}$ such that either $\phi_k$ is a simulating-scenario
of $\psi$, or $\phi_k$ is a shifted simulating-scenario of $\psi$.
As $\rho_{n_0}(a_{0,j_0}^{0})=\min_k\{\rho_{n_0}(s_k^1):s_k^1\in
N(a_{0,j_0}^{0})\}$, it follows directly that $\rho_{n_0}(s_k^1)\ge
\rho^*$. If $\rho(\phi_k)=\rho_{n_0}(s_k^1)\ge \rho^*$, then we stop
and it can be easily seen that the instant approximation ratio of
$\psi$ is at least $\rho^*-O(\epsilon)$ (by Lemma
\ref{le:approx-simu-scenario}). Otherwise we go on to release jobs.

Suppose after releasing $h-1$ jobs the current scenario is $\psi$
and $\phi_i$ is its simulating-scenario or shifted
simulating-scenario, furthermore, $\rho_{n_0}(s_i^{h-1})\ge \rho^*$.
As
$\rho^*\le\rho_{n_0}(s_i^{h-1})=\max_j\{\rho_{n_0}(a_{ij}^{h-1}):a_{ij}^{0}\in
N(s_i^{h-1})\}$, thus there exists some $j_0$ such that
$a_{ij_0}^{h-1}\in N(s_i^{h-1})$ and $\rho_{n_0}(a_{ij_0}^{0})\ge
\rho^*$.

We release the $h$-th job with processing time $\alpha_{j_0}$. Again
suppose this job is scheduled onto a certain machine so that the
scenario becomes $\psi'$, then there exists some $s_k^{h}$ incident
to $a_{ij_0}^{h-1}$ such that $\phi_k$ is either a
simulating-scenario or a shifted simulating-scenario of $\psi'$. As
$\rho_{n_0}(a_{ij_0}^{h-1})=\min_k\{\rho_{n_0}(s_k^{h}):s_k^h\in
N(a_{ij_0}^{0})\}$, it follows directly that $\rho_{n_0}(s_k^h)\ge
\rho^*$. If $\rho(\phi_k)=\rho_{n_0}(s_k^h)\ge \rho^*$, then we stop
and it can be easily seen that the instant approximation ratio of
$\psi'$ is at least $\rho^*-O(\epsilon)$. Otherwise we go on to
release jobs.

Since $\rho(\phi_k)=\rho_{n_0}(s_k^{n_0})$, we stop after releasing
at most $n_0$ jobs.

\subsection{The nearly optimal online algorithm}
\label{ap:strategy sche} We play the part of the scheduler.

Notice that
$$\rho_{n_0+1}(a_{ij}^{0})=\min_k\{\rho_{n_0+1}(s_k^1):s_k^1\in N(a_{ij}^{0})\}=\min_k\{\rho(s_k^1):s_k^1\in N(a_{ij}^{0})\},$$
$$\rho(s_i^0)=\rho_{n_0+1}(s_i^{0})=\max_j\{\rho_{n_0+1}(a_{ij}^{0}):a_{ij}^{0}\in N(s_i^{0})\}.$$

Suppose the current scenario is $\psi$ with scaling factor $T$. Let
$\phi_i\in \Phi'$ be its simulating-scenario or shifted
simulating-scenario, and furthermore, $\rho(s_i^0)\le \rho^*$.

Let $p_n$ be the next job the adversary releases. We apply lazy
scheduling first, i.e., if by scheduling $p_n$ onto any machine,
$\psi$ changes to $\psi'$ (the scaling factor does not change) while
$\phi_i$ is still a simulating-scenario or shifted
simulating-scenario of $\psi'$, we always schedule $p_n$ onto this
machine.

Otherwise, According to Theorem \ref{th:simulate} and Lemma
\ref{le:delete trimmed-scenario}, $p_n'(h)$ could be constructed
such that if $\psi$ changes to $\psi'$ by adding $p_n$ to machine
$h$, then $\phi$ changes to $\phi'$ by adding $p_n'$ to the same
machine such that $\phi'$ is a simulating-scenario or shifted
simulating-scenario of $\psi'$. Notice that the processing time of
$p_n'(h)$ may also depend on the machine $h$.

We show that, if $p_n$ could not be scheduled due to lazy
scheduling, then $p_n'(h)=p_n'$ for every $h$. To see why, we check
the proofs of Theorem \ref{th:simulate} and Lemma \ref{le:delete
trimmed-scenario}. We observe that, if $p_n'(h)\ge
(1+\epsilon)^{-c_0+1}$ for some $h$, then $p_n'(h)=p_n'$ for every
$h$ (the processing time $p_n'(h)$ only depends on $p_n/T$).
Otherwise, it might be possible that $p_n'(h_1)=0$ for some $h_1$
while $p_n'(h_2)=(1+\epsilon)^{-c_0}$ for another $h_2$. However, if
this is the case then $p_n$ should be scheduled on machine $h_1$
according to lazy scheduling, which is a contradiction. Thus,
$p_n'(h)=(1+\epsilon)^{-c_0}$ for every $h$.

Now we decide according to $G_{n_0+1}$ which machine $p_n$ should be
put onto.

As $p_n'\in R$, let $\alpha_{j_0}=p_n'$, then we consider
$\rho_{n_0+1}(a_{ij_0}^{0})=\min_k\{\rho_{n_0+1}(s_k^1):s_k^1\in
N(a_{ij_0}^{0})\}$. Recall that $\rho(s_i^0)\le \rho^*$ according to
the hypothesis, then $\rho_{n_0+1}(a_{ij_0}^{0})\le \rho^*$, which
implies that there exists some $s_{k_0}^1$ incident to $a_{ij_0}$
such that $\rho_{n_0+1}(s_{k_0}^1)=\rho(s_{k_0}^0)\ge \rho^*$. Thus,
we can schedule $p_n'$ to a certain machine, say, machine $h_0$, so
that $\phi_i$ transforms to $\phi_{k_0}$. And thus in the real
schedule we schedule $p_n$ onto machine $h_0$. Let $\psi'$ be the
current scenario, then $\phi_{k_0}$ is its simulating-scenario or
shifted simulating-scenario with $\rho(s_{k_0}^0)\le \rho^*$.

Thus, we can always carry on the above procedure. Since the instant
approximation ratio of each simulating-scenario or shifted
simulating-scenario is no greater than $\rho^*$, the instant
approximation ratio of the corresponding scenario is also no greater
than $\rho^*+O(\epsilon)$.

\section{Extensions}
\label{ap:extension} We show in this section that our method could
be extended to provide approximation schemes for various problems.
Specifically, we consider $Rm||C_{max}$, $Rm||\sum_hC_h^p$ for some
constant $p\ge 1$ (and as a consequence $Qm||C_{max}$,
$Qm||\sum_hC_h^p$ and $Pm||\sum_hC_h^p$ could also be solved). We
mention that, if we restrict that the number of machines $m$ is a
constant (as in the case $Rm||C_{max}$ and $Rm||\sum_hC_h^p$), then
our method could be simplified.

We also consider the semi-online model $P|p_j\le q|C_{max}$ where
the processing time of each job released is at most $q$. In this
case an optimal algorithm could be derived in $(mq)^{O(mq)}$ time.
Notice that our previous discussions focus on finding nearly optimal
online algorithms, however, for online problems, we do not know much
about optimal algorithms. Only the special cases $P2||C_{max}$ and
$P3||C_{max}$ are known to admit optimal algorithms. Unlike the
corresponding offline problems which always admit exact algorithms
(sometimes with exponential running times), we do not know whether
there exists such an algorithm for online problems. Consider the
following problem, does there exist an algorithm which determines
whether there exists an online algorithm for $P||C_{max}$ whose
competitive ratio is no greater than $\rho$. We do not know which
complexity class this problem belongs to. An exact algorithm, even
with running time exponential in the input size, would be of great
interest.

\noindent\textbf{Related work.} For the objective of minimizing the
makespan on related and unrelated machines, the best known results
are in table~\ref{ta:load}. There is a huge gap between the upper
bound and lower bound except for the special case $Q_2||C_{max}$.
However, the standard technique for $Q_2||C_{max}$ becomes extremely
complicated and can hardly be extended for $3$ or more machines.

For the objective of $(\sum_h C_h^p)^{1/p}$, i.e., the $L_p$ norm,
not much is known. See table~\ref{ta:load} for an overview. We
further mention that when $p=2$, List Scheduling is of competitive
ratio $\sqrt{4/3}$~\cite{AvAzSg01}.

\begin{table}[!h]
\begin{center}
\caption{Lower and upper bounds on the competitive ratio for
deterministic}\label{ta:load}
\begin{tabular}{|c|c|c|}
  \hline
  problems & lower bounds & upper bounds \\ \hline
  $Q||C_{\max}$ & 2.564 ~\cite{esgallWAOA11} & 5.828 ~\cite{BeChKa00} \\
  $Q2||C_{\max}$ & $(2s+1)/(s+1)$ for $s \le 1.61803$ & $(2s+1)/(s+1)$ for $s \le 1.61803$, \\
                 & $1+ 1/s$ for $s \ge 1.61803$~\cite{ChoSah80} & $1+ 1/s$ for $s \ge 1.61803$~\cite{ChoSah80}\\
  $R||C_{\max}$ & $\Omega(\log m)$~\cite{azar1992assign} & $O(\log m)$~\cite{aspnes1997line} \\   \hline
  $P||(\sum_h C_h^p)^{1/p}$ & & $2 - \Theta(\ln p/p)$~\cite{AvAzSg01} \\
  $R||(\sum_h C_h^p)^{1/p}$ & & $O(p)$~\cite{awerbuch1995load} \\ \hline
\end{tabular}
\end{center}
\end{table}

Much of the previous work is directed for semi-online models of
scheduling problems where part of the future information is known
beforehand, and most of them assume that the total processing time
of jobs (instead of the largest job) is known. For such a model, the
best known upper bound is 1.6~\cite{cheng2005semi} and the best
known lower bound is 1.585~\cite{albers2012semi}.

\subsection{$Rm||C_{max}$}
In this case, we can restrict beforehand that the processing time of
each job, say, $j$, on machine $h$ ($1\le h\le m$) is
$p_{jh}\in\{(1+\epsilon)^k:k\ge0,k\in \mathbb{N}\}$. There is a
naive algorithm $Al_0$ that puts every job on the machine with the
least processing time, and it can be easily seen that the
competitive ratio of this algorithm is $m$. Since $m$ is a constant,
it is a constant competitive ratio online algorithm, and thus we may
restrict on the algorithms whose competitive ratio is no greater
than $m$.

Given any real schedule, we may first compute the makespan of the
schedule by applying $Al_0$ on the instance and let it be
$Al_0(C_{max})$, then we define $LB=Al_0(C_{max})/m$ and find a
scaling factor $T\in SC$ such that $T\le LB<T(1+\epsilon)^{\omega}$.
Similarly as we do in the previous sections, we can then define a
state for each machine of the real schedule with respect to $T$ and
then a scenario by combining the $m$ states. Since $OPT\le
mT(1+\epsilon)^{\omega}$, if the real schedule is produced by an
online algorithm whose competitive ratio is no greater than $m$,
then the load of each machine is bounded by
$m^2T(1+\epsilon)^{\omega}$, and this allows us to bound the number
of different feasible states by some constant, and the number of all
different feasible scenarios is also bounded by a constant
(depending on $m$ and $1/\epsilon$).

We can then define trimmed-states and trimmed-scenarios in the same
way as before. Specifically, a trimmed-state is combined of $m$
trimmed-states directly (it is much simpler since the number of
machines is a constant). Again, a feasible trimmed-state is a
trimmed-state whose load could be slightly larger than
$m^2T(1+\epsilon)^{\omega}$ (to include two additional small jobs),
and a feasible trimmed-scenario is a trimmed-scenario such that
every trimmed-state is feasible.

Transformations between scenarios and trimmed-scenarios are exactly
the same as before and we can also construct a graph to characterize
the transformations between trimmed-scenarios, and use it to
approximately characterize the transformation between scenarios. All
the subsequent arguments are the same.

\subsection{$Rm||\sum_{h}C_h^p$ when $p\ge 1$ is a constant}
Here $C_h$ denotes the load of machine $h$.

Again we can restrict beforehand that the processing time of each
job, say, $j$, on machine $h$ ($1\le h\le m$) is
$p_{jh}\in\{(1+\epsilon)^k:k\ge0,k\in \mathbb{N}\}$. Consider the
naive algorithm $Al_0$ that puts every job on the machine with the
least processing time and let $C_h(Al_0)$ be the load of machine $h$
due to this algorithm. Since $x^p$ is a convex function, we know
directly that $OPT\ge m(\frac{\sum_{h=1}^mC_{h}(Al_0)}{m})^{p}\ge
\frac{\sum_{h=1}^mC_h(Al_0)^p}{m}$ and thus the competitive ratio of
$Al_0$ is also $m$ and again we may restrict on the algorithms whose
competitive ratio is no greater than $m$.

Given any real schedule, we may first compute the objective function
of the schedule by applying $Al_0$ on the instance and let it be
$Al_0(\sum_{h}C_h^p)$, then we define
$LB=[Al_0(\sum_{h}C_h^p)/m]^{1/p}$ and find a scaling factor $T\in
SC$ such that $T\le LB<T(1+\epsilon)^{\omega}$. Consider any
schedule produced by an online algorithm whose competitive ratio is
no greater than $m$, then its objective value should be bounded by
$mAl_0(\sum_{h}C_h^p)$, which implies that the load of each machine
in this schedule is bounded by
$[mAl_0(\sum_{h}C_h^p)]^{1/p}=m^{2/p}LB$. Again using the fact that
$m$ is a constant, we can then define a state for each machine of
the real schedule with respect to $T$ and then a scenario by
combining the $m$ states. Trimmed-states and trimmed-scenarios are
defined similarly, all the subsequent arguments are the same as the
previous subsection.

\noindent\textbf{Remark. } Our method, however, could not be
extended in a direct way to solve the more general model
$Rm||\sum_{h}f(C_h)$ if the function $f$ fails to satisfy the
property that $f(ka)/f(kb)=f(a)/f(b)$ for any $k>0$. This is because
we neglect the scaling factor when we construct the graph $G$ and
compute the instant approximation ratio for each trimmed-scenario.
Indeed, the instant approximation ratio is not dependent on the
scaling factor for all the objective functions (i.e., $C_{max}$ and
$\sum_h C_h^p$) we consider before, however, if such a property is
not satisfied, then the instant approximation ratio depends on the
scaling factor and our method fails.

\subsection{$P|p_{j}\le q|C_{max}$}

We show in this subsection that, the semi-online scheduling problem
$P|p_j\le q| C_{max}$ in which the largest job is bounded by some
integer $\zeta$ (the value $q$ is known beforehand), admits an exact
online algorithm.

Again we use the previous framework to solve this problem. The key
observation is that, in such a semi-online model, we can restrict
our attentions only on bounded instances in which the total
processing time of all the jobs released by the adversary is bounded
by $2m\zeta$. It is easy to verify that, if we only consider bounded
instances, then we can always use a $\zeta$-tuple to represent the
jobs scheduled on each machine. This is the state for a machine and
there are at most $(2mq)^{q}$ different states. Combining the $m$
states generates scenarios, and there are at most $(2mq)^{mq}$
different scenarios, and thus we can construct a graph to represent
the transformations between these scenarios and find the optimal
online algorithm using the same arguments.

We prove the above observation in the following part of this
subsection.

We restrict that $q\ge 2$ since we assume that the processing time
of each job is some integer, and $q=1$ would implies that the
adversary only releases jobs of processing time $1$, and list
scheduling is the optimal algorithm.

When $q\ge 2$, we know that the competitive ratio of any online
algorithm is no less than $1.5$. To see why, suppose there are only
two machines and the adversary releases at first two jobs, both of
processing time $1$. Any online algorithm that puts the two jobs on
the same machine would have a competitive ratio at least $2$.
Otherwise suppose an online algorithm puts the two jobs on separate
machines, then the adversary releases a job of processing time $2$,
and it can be easily seen that the competitive ratio of this online
algorithm is at least $1.5$.

We use $I$ to denote a list of jobs released by the adversary (one
by one due to the sequence), and this is an instance. We use $LD(I)$
to denote the total processing time of jobs in $I$. Let $\Omega$ be
the set of all instances and $\Omega_B=\{I|LD(I)/m\le 2p\}$ be the
set of bounded instances. Let $A$ be the set of all the online
algorithms. Let $Al\in A$ be any online algorithm, it can be easily
seen that its competitive ratio $\rho_{Al}$ is defined as
$$\rho_{Al}=\sup_{I\in \Omega}\frac{Al(I)}{OPT(I)},$$
where $OPT(I)$ is the makespan of the optimal (offline) solution for
the instance $I$ and $Al(I)$ is the makespan of the solution
produced by the algorithm.

The goal of this subsection is to find an algorithm $Al^*$ such that
$$\rho_{Al^*}=\inf_{Al\in A}\sup_{I\in \Omega}\frac{Al(I)}{OPT(I)}.$$

On the other hand, according to our previous discussion, we can find
an algorithm $Al_B^*$ such that
$$\rho_{Al_B^*}=\inf_{Al\in A}\sup_{I\in \Omega_B}\frac{Al(I)}{OPT(I)}.$$

Notice that when we restrict our attentions on bounded instances,
the algorithm we find may be only defined for $I\in \Omega_B$, we
extend it to solve all the instances in the following way. We use
$LS$ to denote the list scheduling. Given any algorithm $Al$ which
can produce a solution for any instance $I\in \Omega_B$, we use
$Al\circ LS$ to denote the LS-composition of this algorithm where
the algorithm $Al\circ LS$ operates in the following way.

Recall that $I\in \Omega$ is a list of jobs and let it be
$(p_1,p_2,\cdots,p_n)$ where $p_j\ge 1$. If $I\in \Omega_B$, then
$Al\circ LS$ schedules jobs in the same way as $Al$. Otherwise let
$j_0$ be the largest index such that $\sum_{j=1}^{j_0}p_j\le
2m\zeta$, $Al\circ LS$ schedules job $1$ to job $j_0$ in the same
way as $Al$, and schedules the subsequent jobs according to list
scheduling, i.e., when $p_j$ ($j>j_0$) is released, we put this job
onto the machine with the least load currently.

Thus, the algorithm $Al\circ LS$ could be viewed as a combination of
$Al$ and $LS$, and we only require that $Al$ is defined for
instances of $\Omega_B$.

\begin{lemma}
\label{le:circ}For any $Al\in A$,
$$\rho_{Al\circ LS}\le \sup_{I\in \Omega_B}\frac{Al(I)}{OPT(I)}.$$
\end{lemma}
\begin{proof}
Consider $I=(p_1,p_2,\cdots,p_n)\not\in \Omega_B$ and suppose $j_0$
is the largest index such that $\sum_{j=1}^{j_0}p_j\le 2mp$. Let
$I_B=(p_1,p_2,\cdots,p_{j_0})$, then obviously $OPT(I)\ge OPT(I_B)$.

Consider $Al\circ LS(I)$. If $Al\circ LS(I)=Al(I_B)$, then obviously
$Al\circ LS(I)/OPT(I)\le Al(I_B)/OPT(I_B)$.

Otherwise $Al\circ LS(I)>Al(I_B)$, and let $h>j_0$ be the job whose
completion time achieves $Al\circ LS(I) $. Since $h$ is scheduled
due to the LS-rule, we know that $Al\circ LS(I)\le LD(I)/m+p_h$.
Notice that $OPT\ge LD(I)/m\ge 2p$, thus $Al\circ LS(I)/OPT(I)\le
1.5$. Thus we have
$$\rho_{Al\circ LS}\le \max\{\sup_{I_B\in \Omega_B}Al(I_B)/OPT(I_B),1.5\}.$$
Recall that we have shown in the previous discussion that
$\sup_{I_B\in \Omega_B}Al(I_B)/OPT(I_B)\ge 1.5$, thus $\rho_{Al\circ
LS}\le \sup_{I\in \Omega_B}\frac{Al(I)}{OPT(I)}$.
\end{proof}

The above lemma shows that $\rho_{Al_B^*\circ LS}\le\rho_{Al_B^*}$.
Meanwhile it is easy to see that $\rho_{Al_B^*\circ
LS}\ge\rho_{Al_B^*}$, thus $\rho_{Al_B^*\circ LS}=\rho_{Al_B^*}$.

We prove in the following part that $\rho_{Al_B^*}=\rho_{Al^*}$, and
thus $Al_B^*\circ LS$ is the best algorithm for the semi-online
problem.

Obviously $\sup_{I\in \Omega}Al(I)/OPT(I)\ge \sup_{I\in
\Omega_B}Al(I)/OPT(I)$, thus $\rho_{Al^*}\ge \rho_{Al_B^*}$.

On the other hand, let $A\circ LS=\{Al\circ LS:Al\in A\}\subset A$,
$$\inf_{Al\in A}\rho_{Al}\le \inf_{Al\in A\circ LS}\rho_{Al\circ LS}.$$
According to Lemma \ref{le:circ}, for any $I\in \Omega$,
$$\inf_{Al\in A\circ LS}\rho_{Al\circ LS}\le \inf_{Al\in A\circ
LS}\sup_{I\in \Omega_B}\frac{Al(I)}{OPT(I)}=\inf_{Al\in A}\sup_{I\in
\Omega_B}\frac{Al(I)}{OPT(I)},$$ thus $\rho_{Al^*}\le
\rho_{Al_B^*}$, which implies that $\rho_{Al^*}= \rho_{Al_B^*}$.

\end{appendix}

\end{document}